\documentclass[orivec]{llncs}
\usepackage{llncsdoc}
\usepackage{amsmath,amssymb}
\usepackage{subcaption}

\usepackage[american]{babel}
\usepackage[babel]{csquotes}
\usepackage[maxcitenames=2,maxbibnames=9,backend=bibtex,doi=false,isbn=false,url=false]{biblatex}
\bibliography{references}

\usepackage{hyperref}
\usepackage{cleveref}
\usepackage{algorithm}
\usepackage[noend]{algpseudocode}
\usepackage{enumerate}
\usepackage{versions}
\usepackage{mathtools} 

\includeversion{SHORT}
\excludeversion{LONG}
\excludeversion{INUTILE}
\markversion{TODO}

\DeclarePairedDelimiter\ceil{\lceil}{\rceil}

\newcommand{\bigo}[1]{\mathcal{O}(#1)}
\newcommand{\bigolr}[1]{\mathcal{O}\left(#1\right)}

\newcommand{\pushright}[1]{\ifmeasuring@#1\else\omit\hfill$\displaystyle#1$\fi\ignorespaces}

\begin{document}

\hyphenation{pro-blems}
\pagestyle{headings} 
\addtocmark{The Klee's Measure Kernel Problem}

\mainmatter       
\title{Multivariate analysis for \\
\textsc{Maxima} in High Dimensions}
\author{
 J\'er\'emy Barbay
 \and 
 Javiel Rojas
}

\institute{
 Departmento de Ciencias de la Computaci\'on, Universidad de Chile, Chile\\ 
 \email{\{jeremy, jrojas\}@dcc.uchile.cl}
}


\maketitle

\begin{abstract} 
We study the problem of computing the \textsc{Maxima}  of a set of $n$ $d$-dimensional points. For dimensions 2 and 3, there are algorithms to solve the problem with order-oblivious instance-optimal running time. However, in higher dimensions there is still room for improvements. We present an algorithm sensitive to the structural entropy of the input set, which improves the running time, for large classes of instances, on the best solution for \textsc{Maxima} to date for $d \ge 4$.
\end{abstract}

\section{Introduction}\label{sec:Introduction}
The problem of computing the \textsc{Maxima} of a set $\mathcal{S}$ of points was first formulated in 1974 by \textcite{Kung1975}: a point from $\mathcal{S}$ is called \emph{maximal} if none of the remaining vectors in $\mathcal{S}$ dominates it in every component, and the \textsc{Maxima} of $\mathcal{S}$ (denoted by $M(\mathcal{S})$) is the set of maximal points in $\mathcal{S}$.  
For any constant dimension $d$ this problem can be solved naively in time within $\bigo{n^2}$ by comparing every possible pair of points. \textcite{Kung1975} proposed two different algorithms to solve the problem in two and three dimensions, respectively, running in time within $\bigo{n \log n}$. For higher dimensions, \textcite{Kung1975} presented a divide-and-conquer algorithm running in time within $\bigo{n \log^{d-2} n}$ for dimensions $d>3$, and showed a lower bound of within $\Omega(n \log n)$ for this problem for any dimension $d>2$.

In 1985, \textcite{Kirkpatrick1985} gave the first output-size sensitive algorithm for this problem, running in time within $\bigo{n \log h}$ in the plane, and within  $\bigo{n \log^{d-2}{h}}$ for dimension $d>3$, where $h$ is the size of the \textsc{Maxima}. In 2 and 3 dimensions, this result remained unbeaten for almost 25 years until \textcite{Afshani2009} described an \emph{instance-optimal} algorithm for this problem (i.e., an algorithm whose cost is at most a constant factor from the cost of any other algorithm running on the same input, for every input instance). 

Given this improvement in 2 and 3 dimensions, one natural question is whether the same technique can be applied to higher dimensions in order to improve upon \citeauthor{Kirkpatrick1985}'s results~\cite{Kirkpatrick1985} in dimension $d \ge 4$. We show that the upper bound can indeed be applied to higher dimensions, even though the generalization of the order-oblivious instance-optimality result is still open due to the lack of advanced lower bound techniques in high dimensions. We introduce some basic definitions in \Cref{sec:preliminaries}; and describe our generalization of \citeauthor{Afshani2009}'s upper bound~\cite{Afshani2009} to high dimension in \Cref{sec:maxima}, before discussing the potential for instance-optimality results in high dimension in \Cref{sec:discussion}.

\section{Preliminaries}\label{sec:preliminaries}	
An algorithm is said to be \emph{instance-optimal} if its cost is at most a constant factor from the cost of any other algorithm on the same input, for every input instance.
As such algorithm does not always exist, it is often useful to consider a relaxation of this concept, order-oblivious instance-optimality: an algorithm is said to be  \emph{order-oblivious instance-optimal} if it is instance-optimal on the worst input order possible.

\textcite{Afshani2009} improved on the analysis of \textcite{Kirkpatrick1985} for the computation of the \textsc{Maxima} in dimensions 2 and 3 refining it to order-oblivious instance-optimal algorithms. 
In 2 dimensions, they showed that a simple variant of the output sensitive algorithm originally described by \textcite{Kirkpatrick1985} is, surprisingly, order-oblivious instance-optimal.
For the 3 dimensional case, they described a completely new algorithm and a proof of its order-oblivious instance-optimality.

To refine the analysis, \textcite{Afshani2009} introduced the notion of \emph{structural entropy} of a set  $\mathcal{S}$ of points, a measure of difficulty of the input instances of \textsc{Maxima}, and described two algorithms sensitive to this measure.
To define the structural entropy of a point-set, \textcite{Afshani2009} first introduced the concept of \emph{respectful partition}:

\begin{definition}[Respectful Partition]
	A partition $\Pi$ of a set $\mathcal{S}$ of points in  $\mathbb{R}^d$ into disjoint subsets $\mathcal{S}_1, \ldots,\mathcal{S}_t$ is said to be \emph{respectful} if each subset $\mathcal{S}_k$ is either a singleton or can be enclosed by a $d$-dimensional axis-aligned box $B_k$ whose interior is completely below the \textsc{Maxima} of $\mathcal{S}$. 
\end{definition}
Intuitively, the entropy of a respectful partition is the minimum number of bits required to encode it, and the structural entropy of a point-set is the entropy of a respectful partition with minimal entropy.
Formally, the structural entropy of a point-set is defined as follows:
\begin{definition}[Structural Entropy]
The entropy $\mathcal{H}(\Pi)$ of a partition $\Pi$  of $\mathcal{S}$ into disjoint subsets $\mathcal{S}_1, \ldots,\mathcal{S}_t$ is defined as the value
$\sum_{k=1}^{t}{(|\mathcal{S}_k |/n)\log(n/|\mathcal{S}_k |)}$.
The \emph{structural entropy} $\mathcal{H}(\mathcal{S})$ of the input set $\mathcal{S}$ is the minimum of $\mathcal{H}(\Pi)$ over all respectful partitions $\Pi$ of $\mathcal{S}$.
\end{definition}
%
%
\textcite{Afshani2009} showed that the \textsc{Maxima} of a point-set in dimensions 2 and 3, can be computed in time within $\bigo{n(\mathcal{H}(\mathcal{S}) + 1)}$, and proved a matching lower bound for any algorithm which ignores the input order.

In \Cref{sec:maxima} we generalize to higher dimensions \textcite{Afshani2009}'s upper bound on the computational complexity of \textsc{Maxima} in the worst case over instances of fixed size and structural entropy.
This yields an algorithm sensitive to the structural entropy of the input set, hence improving on the solution described by \textcite{Kirkpatrick1985} for $d \ge 4$. 



\section{Multivariate analysis for the computation of \textsc{Maxima}}\label{sec:maxima}
We present the algorithm \texttt{DPC-Maxima} (where \texttt{DPC} stands for ``Divide, Prune, and Conquer") which computes the \textsc{Maxima} of a $d$-dimensional set $\mathcal{S}$ of points with running time sensitive to the structural entropy of the input set. 

Similar to \citeauthor{Afshani2009}'s algorithm~\cite{Afshani2009}, \texttt{DPC-Maxima} performs several pruning steps removing from $\mathcal{S}$ points that are detected to be dominated, until the remaining set becomes the \textsc{Maxima} of the original set. 
For this, the algorithm repeatedly partitions $\mathcal{S}$ into an increasing number of subsets, and for each partition removes from  $\mathcal{S}$  the subsets that are contained in a box completely below the \textsc{Maxima} of $\mathcal{S}$.  
The following is an outline of \texttt{DPC-Maxima}:

\begin{algorithm}        
	\caption{\texttt{DPC-Maxima}}          
	\label{alg:maximaadpt}                           
	\begin{algorithmic}[1]
		\Require A set $S$ of $n$ points in $\mathbb{R}^d$
		\Ensure The \textsc{Maxima} of $\mathcal{S}$
		\State $j \gets 0$, $\mathcal{M}  \gets \{S\}$
		\While {there is a set in $\mathcal{M} $ with more than one point}\label{step:outerLoop}
		\State \label{step:init} $j \gets j+1$, $r_j \gets 2^{2^j}$, $\mathcal{M} \gets \{\}$
       	\State \label{step:partition}partition $S$ into $r_j$ subsets $S_1, S_2, \ldots, S_{r_j}$ of size at most $n/r_j$
		\For {$i=[1..r_j]$}\label{step:forFilter}
		\State Let $\Gamma_i$ be the minimum axis-aligned enclosing rectangle of $S_i$
		\State  \label{step:filter} \textbf{if} $\Gamma_i$ is below the \textsc{Maxima} of $S$ \textbf{then} remove all points in $S_i$ from $S$		\par 
		\hskip \algorithmicindent 	
		\textbf{else} add $S_i$ to $\mathcal{M} $
		\EndFor
		\EndWhile
		\State  return the points remaining in $S$
	\end{algorithmic}
\end{algorithm}

Note that \texttt{DPC-Maxima} always terminates: due to the restriction on the size of the subsets (in step~\ref{step:partition}), at most $\log\log n$ iterations of the outer loop will be performed for any set $S$ of $n$ points. Besides, the algorithm is correct: only dominated points are removed, and in the last iteration of the outer loop every subset is a singleton, so the points remaining in $\mathcal{S}$ after that iteration are precisely those in the \textsc{Maxima}  of $\mathcal{S}$.

In the rest of the section we analyze the running time of \texttt{DPC-Maxima}. Two main issues need to be addressed: how partitions are chosen in step~\ref{step:partition}, and how the filter steps~\ref{step:forFilter}-\ref{step:filter} are performed efficiently.
Partitions need to be chosen with care: if $B$ is a box in an optimal respectful partition of $\mathcal{S}$, the subsets considered by the algorithm that contain points within $B$ need to be small (compact) enough so they fall completely inside $B$ (and the subsets get pruned), but large enough that only a small number of subsets is needed to prune all the points within $B$. Besides, for a given partition, \texttt{DPC-Maxima} needs to check in small time whether the minimum axis-aligned enclosing rectangle of each subset in the partition is below the \textsc{Maxima} of $S$. 

\pagebreak

We address first the selection of the partition. We show in the following lemma that any point set $S$ can be partitioned into subsets so that the faces of any axis aligned box $B$ intersects a small number of them. In particular, if $B$ is a box in a respectful partition  of $S$, the lemma provides a bound on the number of points in $\mathcal{S} \cap B$ remaining after the filtering in steps~\ref{step:forFilter}-\ref{step:filter} of the algorithm \texttt{DPC-Maxima}.


\begin{lemma}\label{lemma:max_partition}
Let $\mathcal{S}$ be a set of $n$ points in $\mathbb{R}^{d}$, and $r$ be an integer in $[1..n]$.  There is a partition $ \Pi$ of $\mathcal{S}$ into $r$ subsets  such that each subset has at most $n/r$ points, and for any axis aligned $d$-dimensional box $B$:
\begin{enumerate}[i.]
	\item \label{item:subsetsIntersected} $B$ partially intersects (i.e. intersects but does not contain) at most $\bigo{r^{1-1/d}}$  subsets in $\Pi$
	\item \label{item:pointsUnfiltered} There are at most  $\min \{|\mathcal{S} \cap B|, \bigo{n/r^{1/d}}\}$ points within $B$ that belong to subsets in $\Pi$ not fully contained within $B$	
\end{enumerate}
Besides, this partition can be found in time within $\bigo{|\mathcal{S}| \log r}$.
\end{lemma}

\begin{proof}
The partition $\Pi$ can be obtained building a $k$-$d$ tree~\cite{Bentley1975}, using the median to split at each level, and recursing until the number of points within each cell drops to $n/r$ or less.
Built like this, the $k$-$d$ tree will have at most $\log r$ levels and $r$ leaf cells, each containing at most $n/r$ points. 
For each leaf cell in the $k$-$d$ tree, the subset of points contained within it is added to $\Pi$.
If needed, empty subsets are added also to obtain the total $r$ subsets.
This way of proceeding requires running time within $\bigo{|\mathcal{S}|}$ for each level, for a total running time within $\bigo{|\mathcal{S}| \log r}$.

To prove proposition~(\ref{item:subsetsIntersected}), we use the fact that any axis aligned $d$-dimensional box intersects at most $\bigo{r^{1-1/d}}$ leaf cells in a $k$-$d$ tree with $\bigo{r}$ leaf cells~\cite{Bentley1975}.
Since the subsets in $\Pi$ fall within cells in a $k$-$d$ tree with at most $r$ leafs, the result follows. 
Finally, (\ref{item:pointsUnfiltered}) derives from (\ref{item:subsetsIntersected}) and from the fact that every subset has at most $n/r$ points: the points within $B$ belong, in the partition, to subsets that are completely contained within $B$, or partially intersecting $B$; the latest ones can be at most $\bigo{r^{1-1/d}}$ because of (\ref{item:subsetsIntersected}), and contain at most $n/r$ points each, for a total number of points within $\bigo{n/r^{1/d}}$, which completes the proof.
\qed
\end{proof}


Now, we turn the attention to the filtering process in steps~\ref{step:forFilter}-\ref{step:filter} of the algorithm \texttt{DPC-Maxima}.
For this, we introduce a data structure that answers dominance queries over a point set, and that benefits from knowing a priori the number of queries to perform. With it, the subsets in the partition that fall within boxes completely below the \textsc{Maxima} can be detected in small time.

\begin{lemma}\label{lemma:fast_cellDominated}
Given a set $S$ of $n$ points in $\mathbb{R}^{d}$, and an integer parameter $r \in [1..n]$, there is a data structure that answers $r$ dominance queries over $S$ (i.e. given a query point $q$, detecting whether exists a point $p \in \mathcal{S}$ that dominates $q$) in total time and space within $\bigo{n \log^{d-2}{r}}$.
\end{lemma}

\begin{proof}
\textcite{Afshani2009} introduced a simple data structure that answers dominance queries over a point-set of size $n$ in time within $\bigo{\log^{d-2} n}$, with preprocessing time within $\bigo{n \log^{d-2} n}$.
Combining this data structure with a grouping trick (e.g. such as described by \textcite{Chan95}) yields the result in the lemma: 
$S$ is partitioned into $\ceil{n/r}$ subsets $S_1,\ldots,S_{\ceil{n/r}}$ of size at most $r$, and each $S_i$ is processed into a data structure to answer dominance queries over $S_i$.  The total preprocessing time is within $\bigo{\frac{n}{r}(r \log ^{d-2} r)} = \bigo{n \log ^{d-2} r}$. 
To answer whether a point $q$ is dominated by any of the points in $S$, the data structures corresponding to each $S_i$ are queried, for $i \in [1..\ceil{n/r}]$: the final answer is ``yes" only if $q$ is dominated in any of the data structures. Thus, each query is answered in time within $\bigo{\frac{n}{r}\log ^{d-2} r}$, for a total time within  $\bigo{n \log^{d-2} {r}}$.
\qed
\end{proof}

Note that, to detect whether a subset of $\mathcal{S}$ falls within a box that is completely below the \textsc{Maxima} of $\mathcal{S}$ (as in step~\ref{step:filter} of \texttt{DPC-Maxima}), it is enough to query whether the corner with maximal coordinates (i.e. the ``upper, right,\ldots" corner) of the minimum axis-aligned rectangle of the subset is dominated by any point in $\mathcal{S}$. 
With this observation, and lemmas \ref{lemma:max_partition} and \ref{lemma:fast_cellDominated} we provide a first bound for the running time of \texttt{DPC-Maxima}:

\begin{lemma}\label{lemma:firstBound}
Let $\mathcal{S}$ be a set of points,  and $\sigma_j$ be the size of $\mathcal{S}$ right after iteration $j$ of the algorithm  \texttt{DPC-Maxima}. The running time of \texttt{DPC-Maxima} is within $\bigolr{\sum_{j=0}^{\log\log n} {\sigma_j \times 2^{j(d-2)}}}$.
\end{lemma}
\begin{proof}
Let $r_j=2^{2^j}$ (as defined in step \ref{step:init} of \texttt{DPC-Maxima}). Obtaining the partition in step~\ref{step:partition} using \Cref{lemma:max_partition} costs time within $\bigo{|\mathcal{S}| \log r_j}$. Besides,
the data structure described in \Cref{lemma:fast_cellDominated} allows to perform the filtering steps in lines \ref{step:forFilter}-\ref{step:filter} collectively in total time within $\bigo{\sigma_j \log^{d-2} r_j}$ (using $r_j$ as the parameter in the lemma).
Applying the definition of $r_j$  yields a total time within $\bigo{\sigma_j \times 2^{j(d-2)}}$. Since at most $\log\log n$ iterations of the outer loop are performed, the result follows.
\qed
\end{proof}
%
 Finally, we show that \texttt{DPC-Maxima} runs in time sensitive to the structural entropy of the input set in the following theorem:

\begin{theorem}\label{theo:maximaHd_runtime}
Let $\mathcal{S}$ be a set of $n$ points in $\mathbb{R}^{d}$, and let $\Pi$ be any optimal respectful partition of $\mathcal{S}$ into subsets ${\mathcal{S}_1,\ldots, \mathcal{S}_h}$ of sizes $n_1,\ldots, n_h$, respectively.
The algorithm \texttt{DPC-Maxima} computes the \textsc{Maxima} of $\mathcal{S}$ in time within $\bigo{n+ \sum_{k=1}^{h}{n_k \log^{d-2}\frac{n}{n_k}}}$.
\end{theorem}

\begin{proof}
Let $\sigma_j$ be the size of $\mathcal{S}$ during iteration $j$.  From \Cref{lemma:firstBound} we have that the running time of \texttt{DPC-Maxima} is within $\bigolr{\sum_{j=0}^{\log\log n} {\sigma_j \times 2^{j(d-2)}}}$.
Now, consider a subset $\mathcal{S}_k$ in $\Pi$, and let $B_k$ be the minimum axis-aligned box enclosing $\mathcal{S}_k$. 
By \Cref{lemma:max_partition}.(\ref{item:pointsUnfiltered}) there are at most $\min \{n_k, \bigo{n/2^{2^j/d}}\}$ points in $\mathcal{S}_k$ that remain in $\mathcal{S}$ after iteration $j$ (remember that $r_j=2^{2^j}$ is used as the parameter in the lemma). By summing over all the subsets in the partition, and plugging the previous bound for the point remaining in each subset after the $j$-th iteration of the outer loop, we have that the running time of the algorithm is within:

\begin{align*}
\sum_{j=0}^{\log\log n} & {\bigo{\sigma_j \times 2^{j(d-2)}}} &\\ 
&\subseteq \sum_{j=0}^{\log\log n}{\sum_{k=1}^{h}{\min \{n_k, \bigo{n/2^{2^j/d}}\} \times 2^{j(d-2)}}}\\
	&= \sum_{k=1}^{h}{\sum_{j=0}^{\log\log n}{\min \{n_k, \bigo{n/2^{2^j/d}}\} \times 2^{j(d-2)}}} %
	\;\;\;\;\;\;\;\;\;\;\;%
	\text{(reordering terms)}\\
	&\in \sum_{k=1}^{h}{\bigolr{\sum_{j=0}^{\log(d \log{\frac{n}{n_k}})}{n_k\times 2^{j(d-2)}} 
		+ \sum_{j =1+\log(d\log{\frac{n}{n_k}})}^{\log\log n}{\frac{n}{2^{\frac{1}{d} {2^{j}}}}\times 2^{j(d-2)}}}} \\
\end{align*}
The left inner summation can be bounded by directly by $\bigo{n_k\log^{d-2}\frac{n}{n_k}}$ \footnote{Use the identity $\sum_{i=0}^{m}{2^{c \times i}}= \frac{2^{c+cm}-1}{2^c-1} \in \bigo{2^{cm}} \subseteq \bigo{\log^c m}$, for any constant $c$.}. To show that the right one  is within $\bigo{n_k}$, start by bounding $2^{j(d-2)}$ by  $2^{\frac{2^{j-1}}{d} }$, and take the infinite sum:
\begin{align*}
	 \sum_{j =1+\log(d\log{\frac{n}{n_k}})}^{\log\log n}  {\frac{n}{2^{\frac{2^{j}}{d}}}\times 2^{j(d-2)}}  
 	\le \sum_{j =1+\log(d\log{\frac{n}{n_k}})}^{\infty}{\frac{n}{2^{\frac{2^{j-1}}{d}}}} 
\end{align*}

This settles the case when $n_k > n/2$: when $n$ is factored out, the remaining series converges, and the total summation is within $\bigo{n} \subseteq \bigo{n_k}$. For $n_k \le n/2$, we substitute the variable of the sum by $i=j - \log(d \log \frac{n}{n_k}) - 1$ to obtain:
\begin{align*}
\sum_{j =1+\log(d\log{\frac{n}{n_k}})}^{\infty}&{\frac{n}{2^{\frac{2^{j-1}}{d}}}}  
= \sum_{i =0}^{\infty}{\frac{n}{2^{2^i \log n/n_k}}} 
= \sum_{i =0}^{\infty}{\frac{n}{{n/n_k}^{2^i}}} \\
	& = n_k \sum_{i=0}^{\infty} {\frac{1}{(n/n_k)^{2^i - 1}}} \le n_k \sum_{i=0}^{\infty} {\frac{1}{(2)^{2^i - 1}}} & \text{(because $n/n_k \ge 2 $)} \\
	& \subseteq \bigo{n_k} &\text{(the serie converges)}
\end{align*}
Finally, replacing the bounds for the two inner summations yields the bound $\bigolr{\sum_{k=1}^{h}{n_k(1+ \log^{d-2}\frac{n}{n_k})}} \subseteq \bigolr{n+ \sum_{k=1}^{h}{n_k \log^{d-2}\frac{n}{n_k}}}$.
\qed
\end{proof}

We prove next that, thanks to the concavity of the \emph{polylog} function, the bound in \Cref{theo:maximaHd_runtime} is never asymptotically worse than $\bigo{n \log^{d-2} h}$, and hence the running time of \texttt{DPC-Maxima} is never worse than the running time of \citeauthor{Kirkpatrick1985}'s algorithm~\cite{Kirkpatrick1985}:

\begin{lemma}
Let $\mathcal{S}$ be a set of $n$ points in $\mathbb{R}^{d}$, of \textsc{Maxima} of size $h$, and let $\Pi$ be any optimal respectful partition of $\mathcal{S}$ into subsets ${\mathcal{S}_1,\ldots, \mathcal{S}_h}$ of sizes $n_1,\ldots, n_h$, respectively. Then,

\begin{displaymath}
 \bigolr{\sum_{k=1}^{h}{n_k \log^{d-2}\frac{n}{n_k}}} \subseteq \bigolr{n \log^{d-2} h}
 \end{displaymath}
\end{lemma}
\begin{proof}
For any concave function $\varphi$, and $a = \sum_k a_k$, $b=\sum_k b_k$, Gibbs' inequality states that
\begin{displaymath}\label{eq:gibbs}
\sum_i{a_k \varphi \left(\frac{b_k}{a_k} \right) \le a \varphi \left(\frac{b}{a} \right)}.
\end{displaymath}
Since $f(x)=\log^{d-2} x$ is a concave function
\footnote{$f(x)''= d (d-1-\log x)(log^{d-2}x)/x^2 \le 0$ for all $x \ge 2^{d-1}$}, 
choosing $a_k = \frac{n_k}{n}$, $b_k=1$ for $k=[1..h]$ yields:

\begin{displaymath}
\sum_{k=1}^{h}{\frac{n_k}{n} \log^{d-2} \left(\frac{n}{n_k} \right) \le \log^{d-2} h }.
\end{displaymath}
Multiplying both sides of the inequality by $n$ yields the result.
\qed
\end{proof}

The upper bound in \Cref{theo:maximaHd_runtime} properly generalizes \citeauthor{Afshani2009}'s result~\cite{Afshani2009} for dimensions 2 and 3 ($d \in [2,3]$), but the lower bound proving is more tricky to prove: we discuss in the next section.

\section{Discussion}\label{sec:discussion}
The bound of $\bigo{n+ \sum_{k=1}^{h}{n_k \log^{d-2}\frac{n}{n_k}}}$ in \Cref{theo:maximaHd_runtime}  is within $\Theta(n \log^{d-2} h)$ for instances where all the subsets in an optimal respectful partition have equal size $n/h$ (as illustrated in \Cref{fig:worstcase}.a).
However, for instances with optimal unbalanced partitions (and hence with low structural entropy), the improvement in running time of $\texttt{DPC-Maxima}$  over \citeauthor{Kirkpatrick1985}'s algorithm~\cite{Kirkpatrick1985} can be significant.
Consider, for example, an instance where all the points not in the \textsc{Maxima} are dominated by a same unique point (as the one illustrated in \Cref{fig:worstcase}.b), and let $h \in \bigo{n^{1-\varepsilon}}$ be the size of the \textsc{Maxima}: \citeauthor{Kirkpatrick1985}'s algorithm~\cite{Kirkpatrick1985} for this instance runs in time within $\bigo{n \log^{d-2} n}$, while $\texttt{DPC-Maxima}$ runs in time within $\bigo{n}$, linear in the input size.

\begin{figure}
	\begin{center}
		\includegraphics{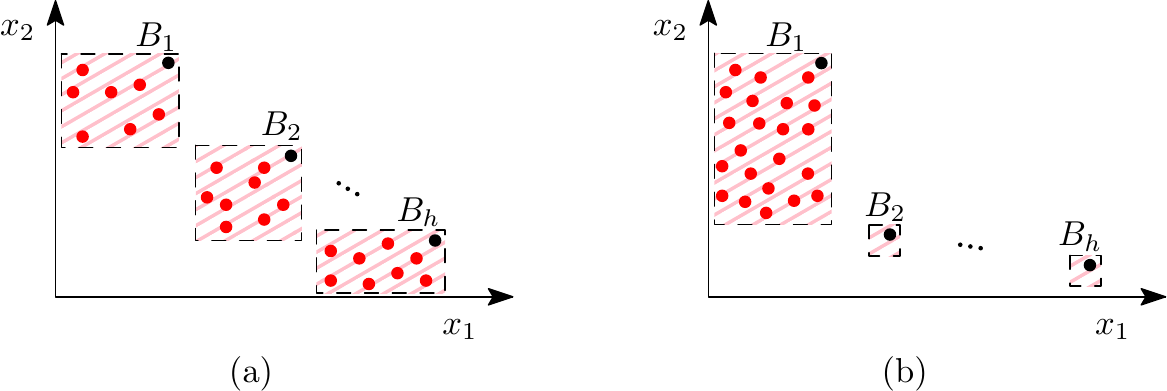}
	\end{center}     
	
	\caption{Two sets of $n$ points and \textsc{Maxima} of size $h$ with worst-case (a) and best-case (b) structural entropies, respectively. In both cases $\Pi = \{B_1,B_2,\ldots,B_h\}$ is an optimal respectful partition, in (a) all the boxes contain $n/h$ points, while in (b) $B_1$ contains $(n-h+1)$ points, and $B_2, \ldots, B_h$ are singletons.}
	\label{fig:worstcase}
\end{figure}

For the problem of \textsc{Maxima} in dimension $d \ge 3$, the best computational complexity lower bound known so far is $\Omega(n(\mathcal{H}(\mathcal{S})+1)) =  \Omega(n+ \sum_{k=1}^{h}{n_k \log \frac{n}{n_k}})$, obtained by extending the bound given in 2009 by \textcite{Afshani2009} for 2 dimensions. 
For dimension $d \ge 4$ there is a gap between this lower bound and the running time of the algorithm \texttt{DPC-Maxima} (which increases with $d$).
The generalization of the order-oblivious instance-optimality result for the $d$-dimensional case for $d \ge 4$ is still open. In general, there is a lack of advanced lower bound techniques for problems handling high dimensional data.
In this sense, a theory of fine classes of problems (as partially done in the field of Parameterized Complexity~\cite{DowneyF99}) could help to palliate this lack of techniques for lower bounds.

A closely related problem is \textsc{Dominance Reporting}, where for a set $\mathcal{S}$ of $n$ points in $d$-dimensional space, and a set of query points $Q$, one must report the points $q \in Q$ for which there is a point in $\mathcal{S}$ that dominates $q$.
For the offline version of the problem, where the size of $Q$ is known to be significantly bigger than the size  of $\mathcal{S}$, the fastest data structure known to date, in dimension $d \ge 3$, was presented in 2012 by \textcite{Afshani2012}. 
\textcite{KarpMR88} considered an online version, where the size of $Q$ is not known in advance, for 
which they describe a deferred data structure sensitive to the sizes of $\mathcal{S}$  and $Q$.
One natural question is whether the results by \textcite{KarpMR88} for online \textsc{Dominance Reporting} 
can be improved via a deferred data structure sensitive to the structural entropy of the input set to query.

\printbibliography

\end{document}